\documentclass[a4paper, english]{lipics-v2018}
\usepackage[T1]{fontenc}
\usepackage[utf8]{inputenc}
\usepackage{csquotes}
\usepackage{microtype}
\usepackage{mathtools}
\usepackage{amsmath}
\usepackage{amssymb}
\usepackage{booktabs}
\usepackage{multirow}
\usepackage{tikz}
\usetikzlibrary{matrix}
 \usetikzlibrary{graphs}

\usetikzlibrary{positioning, arrows.meta, calc}

\usepackage[margin,noinline,draft]{fixme}

\newcommand{\nfold}{$n$-fold}
\usepackage{amsmath}
\usepackage{wrapfig}

\makeatletter
\let\save@mathaccent\mathaccent
\newcommand*\if@single[3]{%
  \setbox0\hbox{${\mathaccent"0362{#1}}^H$}%
  \setbox2\hbox{${\mathaccent"0362{\kern0pt#1}}^H$}%
  \ifdim\ht0=\ht2 #3\else #2\fi
  }
\newcommand*\rel@kern[1]{\kern#1\dimexpr\macc@kerna}
\newcommand*\widebar[1]{\@ifnextchar^{{\wide@bar{#1}{0}}}{\wide@bar{#1}{1}}}
\newcommand*\wide@bar[2]{\if@single{#1}{\wide@bar@{#1}{#2}{1}}{\wide@bar@{#1}{#2}{2}}}
\newcommand*\wide@bar@[3]{%
  \begingroup
  \def\mathaccent##1##2{%
    \let\mathaccent\save@mathaccent
    \if#32 \let\macc@nucleus\first@char \fi
    \setbox\z@\hbox{$\macc@style{\macc@nucleus}_{}$}%
    \setbox\tw@\hbox{$\macc@style{\macc@nucleus}{}_{}$}%
    \dimen@\wd\tw@
    \advance\dimen@-\wd\z@
    \divide\dimen@ 3
    \@tempdima\wd\tw@
    \advance\@tempdima-\scriptspace
    \divide\@tempdima 10
    \advance\dimen@-\@tempdima
    \ifdim\dimen@>\z@ \dimen@0pt\fi
    \rel@kern{0.6}\kern-\dimen@
    \if#31
      \overline{\rel@kern{-0.6}\kern\dimen@\macc@nucleus\rel@kern{0.4}\kern\dimen@}%
      \advance\dimen@0.4\dimexpr\macc@kerna
      \let\final@kern#2%
      \ifdim\dimen@<\z@ \let\final@kern1\fi
      \if\final@kern1 \kern-\dimen@\fi
    \else
      \overline{\rel@kern{-0.6}\kern\dimen@#1}%
    \fi
  }%
  \macc@depth\@ne
  \let\math@bgroup\@empty \let\math@egroup\macc@set@skewchar
  \mathsurround\z@ \frozen@everymath{\mathgroup\macc@group\relax}%
  \macc@set@skewchar\relax
  \let\mathaccentV\macc@nested@a
  \if#31
    \macc@nested@a\relax111{#1}%
  \else
    \def\gobble@till@marker##1\endmarker{}%
    \futurelet\first@char\gobble@till@marker#1\endmarker
    \ifcat\noexpand\first@char A\else
      \def\first@char{}%
    \fi
    \macc@nested@a\relax111{\first@char}%
  \fi
  \endgroup
}
\makeatother

\usepackage{amsthm}
\newtheorem*{lemma*}{Lemma}

\usepackage{xspace}

\newcommand{\Lang}[1]{%
  \ifmmode{%
    \text{\normalfont\textsc{#1}}%
  }%
  \else
  {\normalfont\textsc{#1}}%
  \fi}

\usepackage[smaller]{acronym}
\newacro{EPTAS}{efficient polynomial time approximation scheme}
\newacro{FPTAS}{fully polynomial time approximation scheme}
\newacro{ILP}{integer linear program}
\newacro{LP}{linear program}
\newacro{ETH}{exponential time hypothesis}

\DeclarePairedDelimiter{\norm}{\lVert}{\rVert}

\newcommand{\cO}{\mathcal{O}}

\DeclareMathOperator{\supp}{supp}

\usepackage{cryptocode}

\hideLIPIcs
\nolinenumbers

\authorrunning{K. Jansen, A. Lassota, L. Rohwedder}

\author{Klaus Jansen}{Department of Computer Science, Kiel University,  Kiel, Germany}{kj@informatik.uni-kiel.de}{}{}
\author{Alexandra Lassota}{Department of Computer Science, Kiel University,  Kiel, Germany}{ala@informatik.uni-kiel.de}{}{}
\author{Lars Rohwedder}{Department of Computer Science, Kiel University,  Kiel, Germany}{lro@informatik.uni-kiel.de}{}{}

\title{Near-Linear Time Algorithm for \nfold{} ILPs via Color Coding\footnote{This work was partially supported by DFG Project "Strukturaussagen und deren Anwendung in Scheduling- und Packungsprobleme", JA 612/20-1}}
\begin{document}
\maketitle

\begin{abstract}
We study an important case of ILPs $\max\{c^Tx \ \vert\ \mathcal Ax = b, l \leq x \leq u,\, x \in \mathbb{Z}^{n t} \} $ with $n\cdot t$ variables and lower and upper bounds $\ell, u\in\mathbb Z^{nt}$.
In \nfold{} ILPs non-zero entries only appear in the first $r$ rows of the matrix $\mathcal A$ and in small blocks of size $s\times t$ along the diagonal underneath. Despite this restriction many optimization problems can be expressed in this form. It is known that \nfold{} ILPs can be solved in FPT time
regarding the parameters $s, r,$ and $\Delta$, where $\Delta$ is the greatest absolute value of an entry in $\mathcal A$.
The state-of-the-art technique is a local search algorithm that subsequently moves in an improving direction.
Both, the number of iterations and the search for such an improving direction take time $\Omega(n)$, leading to a quadratic running time in $n$.
We introduce a technique based on Color Coding, which allows us to compute these improving directions in logarithmic time after a single initialization step. 
This leads to the first algorithm for \nfold{} ILPs with a running time that is near-linear in the number $nt$ of variables,
namely
$(rs\Delta)^{\cO(r^2s + s^2)} L^2 \cdot nt \log^{\cO(1)}(nt)$, where $L$ is the encoding length of the largest integer in the input.
In contrast to the algorithms in recent literature,
we do not need to solve the LP relaxation in order to handle unbounded variables.
Instead, we give a structural lemma to introduce appropriate bounds. 
If, on the other hand, we are given such an LP solution, the running time can be decreased by a factor of $L$.
\subjclass{}
\end{abstract}

\section{Introduction}
Solving integer linear programs of the form $\max\,\{c^Tx \ \vert\ \mathcal Ax = b, x\in\mathbb Z_{\ge 0}\}$ is one of the most fundamental tasks in optimization. 
This problem is very general and broadly applicable, but unfortunately also very hard.
In this paper we consider \nfold{} ILPs, a class of integer linear programs with a specific block structure. This is, when non-zero entries appear only in the first $r$ rows of $\mathcal A$ and in blocks of size $s \times t$ along the diagonal underneath.
More precisely, an \nfold{} matrix has the form
\begin{equation*}
\mathcal A =
\begin{pmatrix}
A_1	& A_2	& \dots	& A_n      \\
B_1	& 0 	& \dots  	& 0 	  \\
0	&B_2	&\dots 	& 0	\\
\vdots	& \vdots 	& \ddots & \vdots \\
0 	& 0 & \dots 	 & B_n
\end{pmatrix},
\end{equation*}
where $A_1,\dotsc,A_n$ are $r \times t$ matrices and $B_1,\dotsc,B_n$ are $s \times t$ matrices.
In \nfold{} ILPs we also allow upper and lower bounds on the variables. Throughout the paper we subdivide a solution $x$ into bricks of length $t$ and denote by $x^{(i)}$ the $i$-th one.
The corresponding columns in $\mathcal A$ will be called blocks.

Lately, \nfold{} ILPs received great attention~\cite{altmanova, klein, maack, knop, koutecky} and were studied intensively due to two reasons.
Firstly, many optimization problems are expressible as \nfold{} ILPs~\cite{de2008n, hemmecke2013n, maack, knop}. Secondly, \nfold{} ILPs indeed can be solved much more efficiently than arbitrary ILPs~\cite{klein, hemmecke2013n, koutecky}.
The previously best algorithm has a running time of  $(rs\Delta)^{\cO(r^2s+rs^2)} L \cdot (nt)^2 \log^2(n\cdot t) + \mathrm{LP}$ and is due to Eisenbrand et al.~\cite{klein}. Here $\mathrm{LP}$ is the running time required for solving the corresponding LP relaxation. This augmentation algorithm is the last one in a line of research, where local improvement/augmenting steps are
used to converge to an optimal solution. Clever insights about the structure of the
improving directions allow them to be computed fast.
Nevertheless, the dependence on $n$ in the algorithm above is still high.
Indeed, in practice a quadratic running time is simply not suitable for large data sets \cite{andoni2012approximating, drake2005linear, kelner2014almost}. For example when analyzing big data, large real world graphs as in telecommunication networks or DNA strings in biology, the duration of the computation would go far beyond the scope of an acceptable running time \cite{andoni2012approximating, drake2005linear, kelner2014almost}. For this reason even problems which have an algorithm of quadratic running time are still studied from the viewpoint of approximation algorithms with the objective to obtain results in subquadratic time, even for the cost of a worse quality \cite{andoni2012approximating, drake2005linear, kelner2014almost}. 
Hence, it is an intriguing question, whether the quadratic dependency on the number $nt$ of variables can be eliminated. In this paper, we answer this question affirmatively.
The technical novelty comes from a surprising area: We use a combinatorial structure called splitter, which has been used to derandomize Color Coding algorithms.
It allows us to build a powerful data structure that is maintained during the
local search and from which we can derive an improving direction in logarithmic time.
Handling unbounded variables in an \nfold{} is a non-trivial issue
in the previous algorithms from literature.
They had to solve the corresponding LP relaxation and use proximity results.
Unfortunately, it is not known whether linear programming can be solved in near-linear time in the number
of variables.
Hence, it is an obstacle for obtaining a near-linear running time.
We manage to circumvent the necessity of solving the LP by introducing artificial bounds as a function of the finite upper bounds and the right-hand side of the \nfold{}.
\paragraph*{Summary of Results}
\begin{itemize}
\item We present an algorithm, which solves \nfold{} ILPs in time $(rs\Delta)^{\cO(r^2 s + s^2)} L \cdot nt \log^4(nt) + \mathrm{LP}$, where $\mathrm{LP}$ is the time to solve the LP relaxation of the \nfold{}. This is the first algorithm with a near-linear dependence on the number of variables. The crucial step is to speed up the computation of the improving directions. 
\item We circumvent the need for solving the LP relaxation. This leads to a purely combinatorial algorithm with running time $(rs\Delta)^{\cO(r^2 s + s^2)} L^2 \cdot nt \log^6(nt)$.
\item In the running times above the dependence on the parameters, i.e., $(rs\Delta)^{\cO(r^2 s + s^2)}$, improves
 on the function $(rs\Delta)^{\cO(r^2 s + rs^2)}$ in the previous best algorithms.
\end{itemize}
\paragraph*{Outline of New Techniques}
We will briefly elaborate the main technical novelty in this paper.
Let $x$ be some feasible, non-optimal solution for the \nfold{}.
It is clear that when $y^*$ is an optimal solution for $\max\{c^T y \ | \ \mathcal A y = 0, \ell - x \le y \le u - x, y\in\mathbb Z^{nt}\}$, then $x + y^*$ is optimal for the initial \nfold{}.
In other words, $y^*$ is a particularly good improving step. A sensible approximation
of $y^*$ is to consider  directions $y$ of small size and multiplying them by some step length,
i.e., find some $\lambda \cdot y$ with $\lVert y \rVert_1\le k$ for a value $k$ depending only on $\Delta, r,$ and $s$.
This implies that at most $k$ of the $n$ blocks are used for $y$. If we randomly color the blocks into
$k^2$ colors, then with high probability at most one block of every color is used.
This reduces the problem to choosing a solution of a single brick for every color and
to aggregate them.
We add data structures for every color to implement this efficiently.
There is of course a chance that the colors do not split $y$ perfectly. We handle this
by using a deterministic structure of multiple colorings (instead of one) such that it is
guaranteed that at least one of them has the desired property.

\paragraph*{Related Work} 
The first XP-time algorithm for solving \nfold{} integer programs is due to De Loera et al.~\cite{de2008n} with a running time of  $n^{g(\mathcal A)} L$. Here $g(\mathcal A)$ denotes a so-called Graver complexity of the constraint matrix $\mathcal A$ and $L$ is the encoding length of the largest number in the input. This algorithm already uses the idea of iterative converging to the optimal solution by finding improving directions. Nevertheless, the Graver complexity appears to be huge even for small \nfold{} integer linear programs and thus this algorithm was of no practical use~\cite{hemmecke2013n}. The exponent of this algorithm was then greatly improved by Hemmecke et al. in~\cite{hemmecke2013n} to a constant factor yielding the first cubic time algorithm for solving \nfold{} ILPs. More precisely, the running time of their algorithm is $\Delta^{\cO(t(rs+st))} L\cdot (nt)^3$, i.e., FPT-time parameterized over $\Delta, r, s,$ and $t$.
Lately, two more breakthroughs were obtained. One of the results is due to Kouteck{\'y} et al.~\cite{koutecky}, who gave a strongly polynomial algorithm with running time $\Delta^{\cO(r^2s + rs^2)}(nt)^6 \cdot \log(nt) + LP$. Here $LP$ is the running time for solving the corresponding LP relaxation, which is possible in strongly polynomial time, since the entries of the matrix are bounded. Simultaneously, Eisenbrand et al. reduced in~\cite{klein} the running time from a cubic factor to a quadratic one by introducing new proximity and sensitivity results. This leads to an algorithm with running time $(\Delta r s)^{\cO(r^2s + r s^2)} L \cdot (nt)^2 \log^2(nt) + LP$. Note that both results require only polynomial dependency on $t$.

As for applications, \nfold{} ILPs are broadly used to model various problems.
We refer to the works~\cite{de2008n, hemmecke2013n, DBLP:journals/ol/HemmeckeOW11, maack, DBLP:conf/esa/KnopKM17, DBLP:journals/siamdm/OnnS15} and the references therein for an overview. 

\paragraph*{Structure of the Document}
In Section~\ref{sec-preliminaries} we introduce the necessary preliminaries.
Section~\ref{sec-aug} gives the algorithm for efficiently computing the augmenting steps.
This is then integrated into an algorithm for \nfold{} ILPs in Section~\ref{sec-alg}.
At first we require finite variable bounds and then discuss how to eliminate this requirement
using the solution of the LP relaxation.
Finally, in Section~\ref{sec-bound} we discuss how to handle infinite variable bounds without
the LP relaxation and give new structural results.

\section{Preliminaries}
\label{sec-preliminaries}
In the following we introduce \nfold{}s formally and state the main results regarding them. Further we familiarize splitters, a technique known from Color Coding. 
\begin{definition}
Let $n,r,s,t \in \mathbb N$. Furthermore let $A_1,\dotsc,A_n$ be $r \times t$ integer matrices and $B_1,\dotsc,B_n$ be $s \times t$ integer matrices. Then an \nfold{} $\mathcal A$ is of following form:
\begin{equation*}
\mathcal A =
\begin{pmatrix}
A_1	& A_2	& \dots	& A_n      \\
B_1	& 0 	& \dots  	& 0 	  \\
0	&B_2	&\dots 	& 0	\\
\vdots	& \vdots 	& \ddots & \vdots \\
0 	& 0 & \dots 	 & B_n
\end{pmatrix}.
\end{equation*}
The matrix $\mathcal A$ is of dimension $(r+n\cdot s) \times n\cdot t$. We will divide $\mathcal A$ into blocks of size $(r+n\cdot s) \times t$. Similarly, the variables of a solution $x$ are partitioned into bricks of length $t$.
This means each brick $x^{(i)}$ corresponds to the columns of one submatrix $A_i$ and therefore also $B_i$.
Given $c, \ell, u \in \mathbb{Z}^{n\cdot t}$ and $b \in \mathbb{Z}^{r+n\cdot s}$, the corresponding \nfold{} Integer Linear Programming problem is defined by:
\begin{align*}
 \max\,\{c^Tx \ \vert\ \mathcal Ax = b, \ell \leq x \leq u,\, x \in \mathbb{Z}^{n\cdot t} \} .
\end{align*}
\end{definition}
The main idea for the state-of-the-art algorithms relies on some insight about
the Graver basis of \nfold{}s, which are special elements of the kern of $\mathcal A$.
More formally, we introduce the following definitions:
\begin{definition}
The kern of a matrix $A$ is defined as the set of integral vectors $x$ with $Ax = 0$.
We write $\mathrm{kern}(A)$ for them.
\end{definition}

\begin{definition}
A Graver basis element $g$ is a minimal element of $\mathrm{kern}(A)$.
An element is minimal, if it is not the sum of two sign-compatible elements
$u, v \in \mathrm{kern}(A)$.
\end{definition}
Here, sign-compatible means that $u_i < 0$ if and only if $v_i < 0$ for every $i$.
\begin{theorem}[\cite{cook1986integer}] \label{Graver}
Let $A\in\mathbb Z^{n\times m}$ and
let $x\in\mathrm{kern}(A)$.
Then there exist $2n-1$ Graver basis elements $g_1,\dotsc,g_{2n-1}$, which are sign-compatible with $x$ such that
\begin{align*}
  x = \sum\nolimits_{i=1}^{2n-1} \lambda_i g_i
\end{align*}
for some $\lambda_1,\dotsc,\lambda_{2n-1}\in\mathbb Z_{\ge 0}$.
\end{theorem}
Many results for \nfold{} ILPs rely on the fact that the $\ell_1$-norm of Graver basis elements
for \nfold{} matrices are small.
The best bound known for the $\ell_1$-norm is due to~\cite{klein}.
\begin{theorem} [\cite{klein}] \label{gravernorm}
The $\ell_1$-norm of the Graver basis elements of an \nfold{} matrix $\mathcal A$ is bounded by $\cO(rs\Delta)^{rs}$.
\end{theorem}
Next, we will introduce a technique called splitters (see e.g.~\cite{DBLP:conf/focs/NaorSS95}), which has its origins in the FPT community and was used
to derandomize the Color Coding technique~\cite{alon1995color}.
So far it has not been used with \nfold{} ILPs. We refer the reader to the outline of techniques in the introduction
for the idea on how we apply the splitters.
\begin{definition}
An $(n, k, \ell)$ splitter is a family of hash functions $F$ from $\{1,\dotsc, n\}$ to $\{1,\dotsc,\ell\}$ such that for every $S \subseteq \{1,\dotsc,n\}$ with $|S| = k$, there exists a function
$f \in F$ that splits $S$ evenly,
that is, for every $j, j'\le \ell$ we have $|f^{-1}(j) \cap S|$ and $|f^{-1}(j') \cap S|$ differ by at most 1. 
\end{definition}
If $\ell \ge k$, the above means that there is some hash function that has no collisions when
restricted to $S$.
Interestingly, there exist splitters of very small size.
\begin{theorem}[\cite{alon1995color}] \label{splitComp}
There exists an $(n, k, k^2)$ splitter of size $k^{\cO(1)} \log(n)$ which is computable in time $k^{\cO(1)} \cdot n \log(n)$.
\end{theorem}
We note that an alternative approach to the result above is to use FKS hashing.
Although it has an extra factor of $\log(n)$, it is particularly easy to implement.
\begin{theorem}[Corollary 2 and Lemma 2 in~\cite{DBLP:journals/jacm/FredmanKS84}]
  Define for every prim $q < k^2 \log(n)$ and prim $p < q$ the hash function
  $x \mapsto 1 + (p \cdot (x \text{ mod } q) \text{ mod } k^2)$.
  This is an $(n, k, k^2)$ splitter of size $\cO(k^4\log^2(n))$.
\end{theorem}

\section{Efficient Computation of Improving Directions}
The backbone of our algorithm is the efficient computation of augmenting steps.
The important aspect is the fact that
we can update the augmenting steps very efficiently if the input changes only slightly.
In other words, whenever we change the current solution by applying an augmenting step,
we do not have to recompute the next augmenting step from scratch.
The augmenting steps depend on a partition of the bricks. In the following we define
the notion of a best step based on a fixed partition. Later, we will independently
find steps for a number of partitions and take the best among them.
\label{sec-aug}
\begin{definition}
  Let $P$ be a partition of the $n$ bricks into $k^2$ disjoint sets $P_1, P_2, \dots, P_{k^2}$.
  Let $\overline u\in \mathbb Z^{nt}_{\ge 0}$ and $\overline\ell\in \mathbb Z^{nt}_{\le 0}$ be some 
  upper and lower bounds on the variables (not necessarily the same as in the \nfold{}).
  A $(P, k)$-best step is an optimal solution of
  \begin{align*}
    \max \ & c^T x \\
    \mathcal A x &= 0 \\
    \sum_{i\in B_j} |x_i| &\le k &  \forall j \in \{1, \dots, k^2\} \\
    x_i &= 0 &  \forall j \in \{1, \dots, k^2\}, B_{j'}\in P_j\setminus \{B_j\}, i\in B_{j'} \\
    \overline\ell \le x &\le \overline u \\
    x &\in \mathbb Z^{nt} \\
    B_j &\in P_j & \forall j\in\{1,\dotsc,k^2\} \\
  \end{align*}
\end{definition}
This means a $(P, k)$-best step is an element of kern($\mathcal A$), which uses only one brick
of every $P_j\in P$. Within that brick the norm of the solution must be at most $k$.

\begin{theorem} \label{splitters1}
Consider the problem of finding a $(P, k)$-best step in an \nfold{} where the lower and upper bounds $\overline u, \overline\ell$ can change. This problem can be solved initially in time $k^{\cO(r)} \cdot \Delta^{\cO(r^2 + s^2)} \cdot nt$ and then in $k^{\cO(r)} \cdot \Delta^{\cO(r^2 + s^2)} \cdot \log(nt)$ update time whenever the bounds of a single variable change.
\end{theorem}
\begin{proof}
Let $P$ be a partition of the bricks from matrix $\mathcal A$ into $k^2$ disjoint sets $P_1, P_2, \dots, P_{k^2}$.
Solving the $(P, k)$-best step problem requires that from each set $P_j\in P$ we choose at most one brick and set this brick's variables. All variables in other bricks of $P_j$ must be $0$.

Let $x$ be a $(P, k)$-best step
and let $x^{(j)}$ have the values of $x$ in variables of $P_j$ and $0$
in all other variables. Then by definition, $\lVert x^{(j)} \rVert_1 \le k$.
This implies that the right-hand side regarding $x^{(j)}$, that is to
say, $\mathcal A x^{(j)}$, is also small. Since the absolute value of an entry
in $\mathcal A$ is at most $\Delta$, we have that $\lVert\mathcal A x^{(j)}\rVert_\infty \le k\Delta$. Let $a_i$ be the $i$-th row of $\mathcal A$. If $i > r$, then $a_i x^{(j)} = 0$.
This is because $\mathcal A x = 0$ and $a_i$ has all its support either completely
inside $P_j$ or completely outside $P_j$.
Meaning, the value of $\mathcal A x^{(j)}$ is one of the $(2 k \Delta + 1)^r$ many values we get
by enumerating all possibilities for the first $r$ rows.
Furthermore, since $P$ has only $k^2$ sets,
the partial sum $\mathcal A (x^{(1)} + \cdots + x^{(j)})$ is always one of $(2 k^3 \Delta + 1)^r = (k\Delta)^{\cO(r)}$ many candidates.

Hence to find a $(P, k)$-best step we can restrict our search to solutions whose partial sums stay in this range.
To do so, we set up a graph containing $k^2 + 2$ layers $L_0, L_1, \dots L_{k^2}, L_{k^2+1}$. An example is given in figure \ref{fig:layerGraph}. The first layer $L_0$ will consist of just one node marking the starting point with partial sum zero. Similarly, the last layer $L_{k^2+1}$ will just contain the target point also having partial sum zero, since a $(P, k)$-best step is an element of kern($\mathcal A$).
Each layer $L_j$ with $1 \leq j \leq k^2$ will contain $(2 k^3 \Delta + 1)^{r}$ many nodes, each representing one possible value of $\mathcal A (x^{(1)} + \cdots x^{(j)})$. Two points $v, w$ from adjacent layers $L_{j-1}$, $L_j$ will be connected if the difference of the corresponding partial sums, namely $w-v$, can be obtained by a solution $y$ of variables from only one brick of $P_j$ (with $\lVert y \rVert_1\le k$). The weight of the edge will be the largest gain for the objective function $c^T y$ over all possible bricks. Hence, it could be necessary to compute and compare up to $n$ values for each $P_j$ and each difference in the partial sums to insert one edge into the graph.
Finally, we just have to find the longest path in this graph as it corresponds to a $(P, k)$-best step.
The out-degree of each node is bounded by $(2 k^3 \Delta + 1)^{r}$ since at most this many nodes are reachable in the next layer. Therefore the overall number of edges is bounded by
\begin{align*}
(k^2+2) \cdot (2 k^3 \Delta + 1)^{r} \cdot(2 k^3 \Delta + 1)^{r} = (k \Delta)^{\cO(r)}. 
\end{align*}
Using the Bellman-Ford algorithm we can solve the Longest Path problem for a graph with $N$ vertices and $M$ edges in time $N \cdot M$ as the graph does not contain any circles.
This gives a running time of $(k \Delta)^{\cO(r)} \cdot (k \Delta)^{\cO(r)} = (k \Delta)^{\cO(r)}$
for solving the problem.
Constructing the graph, however, requires solving a number of IPs of the form
\begin{align*}
  \max\ & c^{\prime T} x \\
  \begin{pmatrix} 
    A_j \\
    B_j
  \end{pmatrix} 
  x &= \begin{pmatrix}
  b' \\
  0
  \end{pmatrix} \\
   \lVert x \rVert_1 &\le k \\
   \ell'\le x &\le u' \\
   x &\in\mathbb Z^t ,
\end{align*}
where $b'\in\mathbb Z^r$ is the corresponding right-hand side of the top rows and $\ell', u', c'$ are the upper and lower bounds, and the objective of the block.
This is an IP with $r + s$ constraints, $t$ variables, lower
and upper bounds, and entries of the matrix bounded by $\Delta$ in absolute value.
Using the algorithm by Eisenbrand and Weismantel~\cite{eisenbrand2018proximity},
solving one of them requires time
\begin{equation*}
  t \cdot \cO(r+s+1)^{r + s + 4}\cdot \cO(\Delta)^{(r + s + 1)(r + s + 4)} \cdot \log^2((r+s+1)\Delta) = t \cdot \Delta^{\cO(r^2 + s^2)} .
\end{equation*}
In fact, a little thought allows us to reduce the dependency on $t$ to a logarithmic one:
Since the number of constraints in the ILP above is very small, there are only $\Delta^{\cO(r + s)}$
many different columns. Because of the cardinality constraint $\lVert x \rVert_1 \le k$,
we only have to consider $2k$ many variables of each type of column, namely:
\begin{itemize}
  \item The $k$ many with $u'_i > 0$ and maximal $c'_i$ and 
  \item the $k$ many with $\ell'_i < 0$ and minimal $c'_i$.
\end{itemize}
If some solution uses a variable not in this set, then by pigeonhole principle there
is a variable with the same column values and a superior objective value and
which can be increased/decreased. We can reduce the variable
outside this set and increase the corresponding variable inside this set until all
variables outside the set are $0$.
We can use an appropriate data structure (e.g. AVL trees) to maintain a set of all variables
with $u'_i > 0$ ($\ell'_i < 0$) such that we can find the $k$ best among them in time
$\cO(k\log(t))$. Whenever the bounds of some variable change, we might have to add or remove
entries, which also takes only logarithmic time.
After initialization in time $\cO(nt)$ (in total for all bricks) solving such an IP can therefore be implemented in time
\begin{equation*}
  k\log(t) + 2k\Delta^{\cO(r + s)} \Delta^{\cO(r^2 + s^2)} \le k\log(t) \Delta^{\cO(r^2 + s^2)}.
\end{equation*}
The number of IPs to solve is at most $n$ times the number of edges,
since we have to compare the values of up to $n$ bricks.
This gives a running time of
\begin{equation*}
  \cO(nt) + n \cdot (k\Delta)^{\cO(r)} \cdot \log(t) \cdot \Delta^{\cO(r^2+s^2)} \le nt\cdot k^{\cO(r)} \cdot \Delta^{\cO(r^2+s^2)}
\end{equation*}
for constructing the graph.
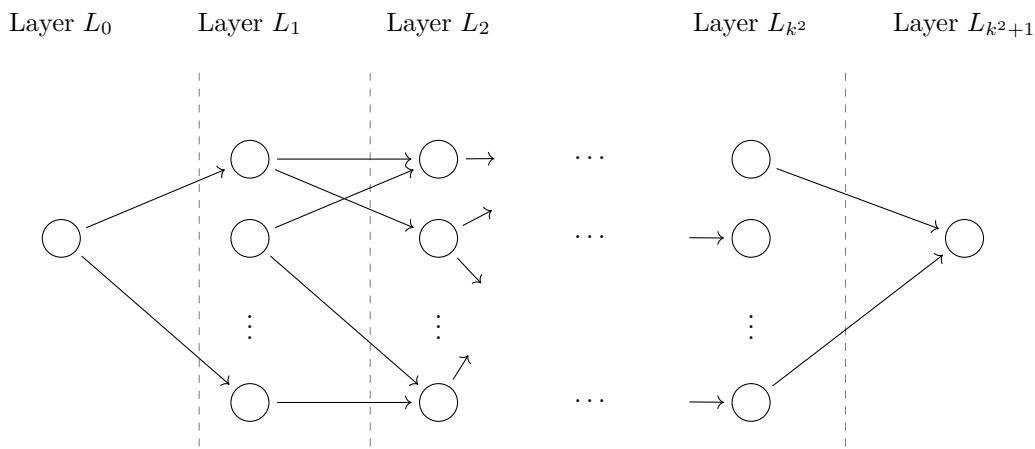
\begin{figure}
\centering
\begin{tikzpicture}
    \matrix (G)[matrix of nodes, column sep=8mm, row sep=3mm,nodes={draw,circle, minimum size=0.5cm}] {
    \node[draw=none]{Layer $L_0$}; & \node[draw=none]{Layer $L_1$}; & \node[draw=none]{Layer $L_2$}; &  & \node[draw=none]{Layer $L_{k^2}$}; & \node[draw=none]{Layer $L_{k^2+1}$}; \\
   			   		& \node (v2) {};						& \node (v5) {};	& \node[draw=none] (d4){$\dots$};	& \node (v8) {};		& \\
      \node (v1) {};	& \node (v3) {};						& \node (v6) {};	&\node[draw=none] (d5){$\dots$};	& \node (v9) {};		& \node (v12) {};\\
                     		& \node[draw=none] (d1){$\vdots$};	& \node[draw=none] (d2){$\vdots$}; &	 &\node[draw=none] (d3){$\vdots$};	& \\
				 	& \node (v4) {};						& \node (v7) {};	& \node[draw=none] (d6){$\dots$};	& \node (v10) {};		& \\     						   
    };
    
     \draw[dashed, gray] (-4.25, 1.5) -- (-4.25,-3.5);
      \draw[dashed, gray] (-2, 1.5) -- (-2,-3.5);
       \draw[dashed, gray] (4.25, 1.5) -- (4.25,-3.5);
    
    \graph [use existing nodes] {
      v1 ->[shorten <= 0.1cm, shorten >= 0.1cm] {v2, v4};
      v2 ->[shorten <= 0.1cm, shorten >= 0.1cm] {v5, v6};
      v3 ->[shorten <= 0.1cm, shorten >= 0.1cm] {v5, v7};
      v4 ->[shorten <= 0.1cm, shorten >= 0.1cm] {v7};
      v8 ->[shorten <= 0.1cm, shorten >= 0.1cm] {v12};
      v10 ->[shorten <= 0.1cm, shorten >= 0.1cm] {v12};
      v5 ->[shorten <= 0.1cm, shorten >=0.9cm] {d4};
      v6 ->[shorten <= 0.1cm, shorten >=1.1cm] {d4};
      v6 ->[shorten <= 0.1cm, shorten >=1.75cm] {d6};
      v7 ->[shorten <= 0.1cm, shorten >=2.7cm] {d4};
      d5 ->[shorten <=0.9cm, shorten >= 0.1cm] {v9};
      d6 ->[shorten <=0.9cm, shorten >= 0.1cm] {v10};
    };
   
\end{tikzpicture}
\caption{This figure shows an example for a layered graph obtained while solving the $(P,k)$-best step problem. There are $k+2$ layers, visually separated by gray dashed lines. This includes one source layer $L_0$, one target layer $L_{k^2+1}$ both with just a single node representing the zero sum. Further there are $k^2$ layers with $(2 k^3 \Delta + 1)^{r}$ nodes each, where in one layer the nodes stand for all reachable partial sums. Two points $v, w$ from adjacent layers $L_{j-1}$, $L_j$ will be connected if the difference of the corresponding partial sums, namely $w-v$, can be obtained by a solution $y$ of variables from only one brick of $P_j$ (with $\lVert y \rVert_1\le k$). The weight of the edge will be the largest gain for the objective function $c^T y$ over all possible bricks. For the sake of clarity both the values of the nodes and the edges are not illustrated.} \label{fig:layerGraph}
\end{figure}
To obtain the update time from the premise of the theorem, it is perfectly fine to solve
the Longest Path problem again, but we cannot construct the graph from scratch.
However, in order to construct the graph we still have to find the best value over all bricks for each edge.
Fortunately, if only a few bricks are updated (in their lower and upper bounds) it is not necessary to recompute all values.
Each edge corresponds to a particular $P_j\in P$ and a fixed right-hand side (a possible value of $\mathcal A x^{(j)}$).
We require an appropriate data structure $\mathcal D_e$ for every edge $e$, which supports fast computation of the operations \textsc{FindMax}, \textsc{Insert}, and \textsc{Delete}. Again, an AVL tree computes each of these operations in time $\cO(\log(N))$, where $N$ is the number of elements.
In $\mathcal D_e$ we store pairs $(v, i)$ where $i$ is a brick in $P_j$ and $v$ is the maximum gain 
of brick $i$ for the right-hand side of $e$. The pairs are stored in lexicographical order.
Since there are at most $n$ bricks in $P_j$, the data structure will have at most $n$ elements.
Initially, we can build $\mathcal D_e$ in time $nt \cdot \Delta^{\cO(r^2 + s^2)}$
(this is replicated for each edge).
Now consider a change to the instance.
Recall that we are looking at changes that affect only a single brick, namely the upper and lower bounds
within that brick change. We are going to update the data structure $\mathcal D_e$ (for each edge) to reflect the changes
and we are going to recompute the edge value of each edge $e$ using $\mathcal D_e$.
Then we simply solve the Longest Path problem again.
Let $P_j\in P$ be the set that contains the brick $i$ that has changed in some variable. We only have to consider edges from $L_{j-1}$ to $L_j$, since none of the other edges are affected by the change.
For a relevant edge $e$ we compute the previous value $v$ and current value $v'$ that the brick $i$
would produce (before and after the bounds have changed). In $\mathcal D_e$ we have to remove $(v, i)$
and insert $(v', i)$. Both operations need only $O(\log(n))$ time.
Then the running time to update $\mathcal D_e$ for one edge is
\begin{equation*}
  k\log(t) \cdot \Delta^{\cO(r^2 + s^2)} + \cO(\log(n))
  \le k\log(nt) \cdot \Delta^{\cO(r^2 + s^2)} .
\end{equation*}
In order to update the edge value of $e$ using $\mathcal D_e$, we simply have to find the maximum element in
$\mathcal D_e$, which again takes time $\cO(\log(n))$.
To summarize, the total time to update the $(P, k)$-best step after a change to a single brick consists of
(1) updating each $\mathcal D_e$, (2) finding the maximum in each $\mathcal D_e$, and (3)
solving the Longest Path problem. We conclude that the update time is
\begin{multline*}
  k\log(nt) \cdot \Delta^{\cO(r^2 + s^2)} \cdot (k\Delta)^{\cO(r)}
  + \log(n) \cdot (k\Delta)^{\cO(r)}
  + (k\Delta)^{\cO(r)} \\
  \le k^{\cO(r)} \Delta^{\cO(r^2 + s^2)} \cdot \log(nt) .
  \qedhere
\end{multline*}
\end{proof}

\section{The Augmenting Step Algorithm}
\label{sec-alg}
In this section we will assume that all lower and upper bounds are finite and give a complete algorithm
for this case. Later, we will explain how to cope with infinite bounds.
We start by showing how to converge to an optimal solution when an initial feasible solution is given.
To compute the initial solution, we also apply this algorithm on a slightly modified instance.
The approach resembles the procedure in previous literature, although we apply the results from
the previous section to speed up the computation of augmenting steps.

Let $x$ be a feasible solution for the \nfold{}, in particular $\mathcal A x = b$. Let $x^*$ be an optimal one. Theorem~\ref{Graver} states that we can decompose the difference vector $x' = x^* - x$ into at most $2nt$ weighted Graver basis elements, that is 
\begin{align*}
x' = x^* - x = \sum\nolimits_{j=1}^{2nt-1} \lambda_j g_j .
\end{align*}
For intuition, consider the following simple approach
(this is similar to the algorithm in~\cite{hemmecke2013n}).
Suppose we are able to guess the best vector $\lambda_i g_i = \mathrm{argmax}_j\{c^T (\lambda_j g_j)\}$
regarding the gain for the objective function.
This pair of step length $\lambda_i$ and Graver element $g_i$ is called the Graver best step. Then we can augment the current solution $x$ by adding $\lambda_i g_i$ to it, i.e., we set $x \leftarrow x + \lambda_i g_i$.
Feasibility follows because all $g_j$ are sign-compatible.
This procedure is repeated until no improving step is possible and therefore $x$ must be optimal.
In each iteration this decreases the gap to the optimal solution by a factor of at least $1-1/(2nt)$
by the pigeonhole principle.
It may be costly to guess the precise Graver best step, but for our purposes it will suffice to
find an augmenting step that is approximately as good.

We will now describe how to guess $\lambda_i$.
Since $x + \lambda_i g_i$ is feasible, we have that $\lambda_i g_i \le u - x \le u - \ell$ and
$\lambda_i g_i \ge \ell - x \ge \ell - u$. Let $(g_i)_j\in \supp(g_i)$ be some non-zero variable.
If $(g_i)_j > 0$, then $\lambda_i \le (\lambda_i g_i)_j \le u_j - \ell_j$.
Otherwise, $(g_i)_j < 0$ and $\lambda_i \le - (\lambda_i g_i)_j \le -(\ell_j - u_j) = u_j - \ell_j$.
Hence, it suffices to check all values in the range $\{1,\dotsc,\Gamma\}$, where
$\Gamma = \max_j\{u_j - \ell_j\}$.
Proceeding like in \cite{klein}, we lower the time a bit further by not taking every value into consideration.
Instead, we look at guesses of the form $\lambda' = 2^k$ for $k \in \{0, \dotsc, \lfloor\log(\Gamma)\rfloor\}$.
Doing so we lose a factor of at most $2$ regarding the improvement of the objective function,
since $c^T (\lambda' g_i) > 0.5 \cdot c^T (\lambda_i g_i)$ when taking
$\lambda'=2^{\lfloor\log(\lambda_i)\rfloor} > \lambda_i / 2$.
Fix $\lambda'$ to the value above.
Next we describe how to compute an augmenting step that is at least as good as $\lambda' g_i$.
Note that $g_i$ is a solution of
\begin{align*}
  \mathcal A y &= 0 \\
  \lVert y \rVert_1 &\le k \\
  \lceil\frac{\ell-x}{\lambda'}\rceil \le y &\le \lfloor\frac{u-x}{\lambda'}\rfloor,
\end{align*}
where $k = \cO(rs\Delta)^{rs}$ is the bound on the norm of Graver elements from Theorem~\ref{gravernorm}.
Suppose we have guessed some partition $P = \{P_1,\dotsc,P_{k^2}\}$ of the bricks such that of each $P_j$ only
a single brick has non-zero variables in $g_i$. Clearly, the augmenting step $\lambda' y^*$, where
$y^*$ is a $(P, k)$-best step with bounds
$\overline\ell = \lceil\frac{\ell-x}{\lambda'}\rceil$ and $\overline u = \lfloor\frac{u-x}{\lambda'}\rfloor$
would be at least as good as $\lambda' g_i$.
Indeed Theorem \ref{splitters1}  explains how to compute such a $(P, k)$-best step dynamically and
when we add $\lambda' g_j$ to $x$ we only change the bounds of at most $k^3$ many variables. Hence,
it is very efficient to recompute $(P, k)$-best steps until we have converged to the optimal solution.
However, valid choices of $\lambda'$ and $P$ might be different in every iteration.
Regarding $\lambda'$, we simply compute $(P, k)$-best steps for every of the $\cO(\log(\Gamma))$ many guesses and
take the best among them. We proceed similarly for $P$. We guess a small number of partitions
and guarantee that always at least one of them is valid. For this purpose we employ splitters.
More precisely, we compute a $(n, k, k^2)$ splitter of the $n$ bricks.
Since $g_j$ has a norm bounded by $k$, it can also only use at most $k$ bricks. Therefore the
splitter always contains a partition $P = \{P_1,\dotsc,P_{k^2}\}$ where $g_j$ only uses a single brick
in every $P_j$.

To recap, in every iteration we solve a $(P, k)$-best step problem
for every guess $\lambda'$ and every partition $P$ in the splitter and
take the overall best solution as an improving direction $\lambda' y^*$.
Then we update our solution $x$ by adding $\lambda' y^*$ onto it.
At most $k^2$ many bricks change (and within each brick only $k$ variables can change)
and therefore we can efficiently recompute the $(P, k)$-best steps
for every guess for the next iteration.
This way we guarantee that we improve the solution by a factor of at least $1-1/(4nt)$ in every
iteration. The explicit running time of these steps will be analyzed in the next theorem.

\paragraph*{Initial Solution}
Recall that we still have to find an initial solution. This solution indeed can be computed by using the augmenting step algorithm described above.
We construct a new \nfold{} ILP which has a trivial feasible solution and whose optimal solution corresponds
to a feasible solution of the original problem.

First we extend our \nfold{} $\mathcal A$ by adding $(r+s)n$ new columns as follows: After the first block $(A_1, B_1, 0, \dots, 0)^T$ add $r+s$ columns. The first $r$ ones will contain an $r \times r$ identity matrix we call $\mathcal I_r$. This matrix $\mathcal I_r$ has all ones in the diagonal. All other entries are zero. The next $s$ columns will contain an $s \times s$ identity matrix $\mathcal I_s$. This submatrix will start at row $r+1$. Again all other entries are zeros in these columns. After the next block we again introduce $r+s$ new columns, the first $r$ ones containing just zeros, the next an $\mathcal I_s$ matrix at the height of $B_2$. We repeat this procedure of adding $r+s$ columns after each block, the first $r$ having solely zero entries and the next $s$ containing $\mathcal I_s$ at the height of $B_i$ until our resulting matrix $\mathcal A^{\text{init}}$ for finding the initial solution looks like the following:
\begin{equation*}
\mathcal A^{\text{init}} =
\begin{pmatrix}
\begin{tikzpicture}
  \matrix (m)[matrix of math nodes, nodes in empty cells, nodes={minimum width = 0.8cm, minimum height = 0.4cm} ]
   {
A_1		& \mathcal I_r & 	0	& A_2 & 0	& 0 		& \dots		& A_n  	& 0 & 0    \\
B_1		& 0 		& \mathcal I_s & 0 	& 0	& 0 		& \dots  		& 0 	 	& 0 & 0 \\
0		& 0 		& 0 & B_2 & 0	& \mathcal I_s 		& \dots 		& 0	& 0	& 0\\
0		& 0 		& 0 & 0 	& 0	& 0 		& \ddots 		& 0 		& 0 & 0\\
\vdots	& \vdots &  \vdots & \vdots & \vdots & 	\vdots	&	& \vdots & \vdots & \vdots \\
0 		& 0	 	& 0 & 0	& 0	& 0	&	& B_n & 0	& \mathcal I_s \\
};

\draw(m-1-1.south west) rectangle (m-1-4.north west);
\draw(m-1-4.north west) rectangle (m-1-7.south west);
\draw(m-1-8.south west) rectangle (m-1-10.north east);
\draw(m-1-1.south west) rectangle (m-2-4.south west);
\draw(m-2-4.south west) rectangle (m-3-6.south east);
\draw(m-6-8.south west) rectangle (m-6-10.north east);
\end{tikzpicture}
\end{pmatrix}.
\end{equation*}
Due to our careful extension $\mathcal A^{\text{init}}$ has again \nfold{} structure. For clarity the relevant submatrices are framed in the matrix above. Remark that zero entries inside of a block do not harm as solely the zeros outside of the blocks are necessary for an \nfold{} structure. 
At first glance, it seems that for the right-hand side $b$ we now have a trivial solution consisting only of the new columns. Keep in mind, however, that the old variables have upper and lower bounds and that $0$ might
be outside these bounds. In order to handle this case we subtract $\ell$, the lower bound, from
all upper and lower bounds and set the right-hand side to $b' = b - \mathcal A \ell$.
We get an equivalent \nfold{} where every solution is shifted by $\ell$.
Now we can find a feasible solution (for $b'$) using solely the new variables by defining 
\begin{align*}
y' = (0, \dots 0, b'_1, b'_2, \dots b'_{r+s}, 0, \dots 0, b'_{r+s+1}, \dots b'_{r+2s}, 0, \dots 0, b'_{r+ns-s}, b'_{r+ns})^T
\end{align*}
where each non-zero entry corresponds to the columns containing the submatrices $\mathcal I_r$ and $\mathcal I_s$ respectively with a multiplicity of the remaining right-hand side $b'$.
Next we introduce an objective function that penalizes using the new columns by
having non-zero entries $c'_i$ corresponding to the positions of the new variables.
We set
\begin{align*}
c^{\text{init}} = (0, \dots 0, c'_1, \dots c'_{r+s}, 0, \dots 0, c'_{r+ns-s}, \dots, c'_{r+ns}),
\end{align*}
where the zero entries correspond to old variables.
The values $c'_i$ and the lower and upper bounds for the new variables depend on the sign of the right-hand side.
\begin{itemize}
\item If $b'_i \geq 0$, then set $c'_i = -1$, the lower bound to $0$, and the upper bound to $b'_i$.
This way the variable can only be non-negative.
\item If $b'_i < 0$, set $c_i = 1$, the lower bound to $b'_i$ and the upper bound to $0$.
Hence this variable must be non-positive.
\end{itemize}
Clearly a solution has a value of $0$, if and only if none of the new columns are used and
no solution of better value is possible.
Hence, if we use our augmenting step algorithm and solve this problem optimally,
we either find a solution with value $0$ or one with a negative value.
In the former, we indeed have not taken any of the new columns into our solution, therefore we can delete the new columns and obtain a solution for the original problem (after adding $\ell$ to it).
Otherwise, there is no feasible solution for the original problem as we solved the problem optimal regarding the objective function.

\begin{theorem}
The dynamic augmenting step algorithm described above computes an optimal solution for the \nfold{} Integer Linear Program problem in time $(rs\Delta)^{\cO(r^2s + s^2)} \cdot  \cO(L^2 \cdot nt \log^4(nt))$ when finite variable bounds are given for each variable. Here $L$ is the encoding length of the largest occurring number in the input.
\end{theorem}
\begin{proof}
Due to Theorem \ref{Graver} we know that the difference vector of an optimal solution $x^*$ to our current solution $x$, i.e. $x' = x^* - x$, can be decomposed into $2nt$ weighted Graver basis elements.
Hence, if we adjust our solution $x$ with the Graver best step, we reduce the gap between the value of an optimal solution and our current solution by a factor of at least $1-1/(2nt)$ due to the pigeonhole principle.
Our algorithm finds an augmenting step that is at least half as good as the Graver best step.
Therefore, the gap to the optimal solution is still reduced by at least a factor of $1-1/(4nt)$.

Regarding the running time we first have to compute the splitter. Theorem \ref{splitComp} says, that this can be done in time $k^{\cO(1)} \cdot n \log(n) = (rs\Delta)^{\cO(rs)} \cdot n \log(n)$.
Next we have to try all values for the weight $\lambda$. Due to our step-length we get $\cO(\log(\Gamma))$ guesses.
Recall that $\Gamma$ denotes the largest difference between an upper bound and the corresponding lower bound, i.e.,  $\Gamma = \max_j\{u_j - \ell_j\}$.
Fixing one, we have to find the best improving direction regarding each of the $((rs\Delta)^{\cO(rs)})^{\cO(1)} \log(n) = (rs\Delta)^{\cO(rs)} \log(n)$ partitions.
In the first iteration we have to set up the tables in time $k^{\cO(r)} \cdot \Delta^{\cO(r^2 + s^2)} \cdot nt = (rs\Delta)^{\cO(r^2s)} \cdot \Delta^{\cO(r^2 + s^2)} \cdot nt$ by computing the gain for each possible summand for each set and setting up the data structure.
In each following iteration we update each table and search for the optimum in time $k^{\cO(r)} \cdot \Delta^{\cO(r^2 + s^2)} \cdot \log(nt) = (rs\Delta)^\cO(r^2s) \cdot \Delta^{\cO(r^2 + s^2)} \cdot \log(nt)$. Now it remains to bound the number $I$ of iterations needed to converge to an optimal solution. To obtain such a bound we calculate:
\begin{align*}
1 & > (1-1/(4nt))^I |c^T(x^* - x)| .
\end{align*}
By reordering the term, we get
\begin{align*}
 I & < \frac{-\log(|c^T(x^* - x)|)}{\log(1-1/(4nt))} .
\end{align*}
As $\log(1+x) = \Theta(x)$, we can bound $\log(1-1/(4nt))$ by $\Theta(-1/(4nt))$ and thus
\begin{align*}
I & < \cO( \frac{-\log(|c^T(x^* - x)|)}{-1/(4nt)} ) \leq \cO(4nt \log(|c^T(x^* - x)|)) .
\end{align*}
As the maximal difference between the current solution $x$ and an optimal one $x^*$ can be at most the maximal value of $c$ times the largest number in between the bounds for each variable, we get $|c^T(x^* - x)| \leq nt \max_i |c_i| \cdot \Gamma$ and thus
\begin{align*}
I & < \cO( 4nt \log(|c^T(x^* - x)|) ) \leq \cO(nt \log(nt \max_i |c_i| \cdot \Gamma)) \leq \cO(nt \log(nt\Gamma \max_i |c_i|)) .
\end{align*}
Let $L$ denote the encoding length of largest integer in the input. Clearly $2^L$ bounds the largest absolute value in $c$ and thus we get
\begin{align*}
I & < \cO(nt \log(nt\Gamma \max_i |c_i|)) =  \cO( nt \log(nt\Gamma 2^{L})) = \cO( nt \log(nt\Gamma 2^L)) .
\end{align*}

Hence after this amount of steps by always improving the gain by a factor of at least $1-1/(4nt)$ we close the gap between the initial solution and an optimal one. Given this, we can now bound the overall running time with:
\begin{align*}
&\underbrace{(rs\Delta)^{\cO(rs)} \cdot n \log(n)}_\text{Splitter} + \underbrace{(rs\Delta)^{\cO(rs)}}_\text{Partitions} \cdot   \underbrace{(rs\Delta)^{\cO(r^2s)} \cdot (rs\Delta)^{\cO(r^2+s^2)} \cdot nt}_\text{First Iteration} + \\
& \underbrace{\cO(nt \log(nt\Gamma 2^L))}_\text{$\mathcal I$}  \cdot \underbrace{\cO(\log(\Gamma))}_\text{$\lambda$ Guesses} \cdot \underbrace{(rs\Delta)^{\cO(rs)}}_\text{Partitions} \cdot \underbrace{(rs\Delta)^{\cO(r^2s)} \cdot (rs\Delta)^{\cO(r^2+s^2)} \cdot \log(nt))}_\text{Update Time} \\
&= \cO((nt \log(nt\Gamma 2^L)) \cdot \cO(\log(\Gamma)) \cdot (rs\Delta)^{\cO(r^2s + s^2)} \cdot \log(nt) \\
& =  (rs\Delta)^{\cO(r^2s + s^2)} \cdot  \cO(\log^2(\Gamma + 2^L) \cdot nt \log^2(nt)) .
\end{align*}
Here \emph{Splitter} denotes the time to compute the initial set $\mathcal P$ of partitions and \emph{Partitions} denotes the  cardinality of $\mathcal P$. \emph{First Iteration} is the time to solve the first iteration of the $(P, k)$-best step problem. Further \emph{$\lambda$ Guesses} is the number of guesses we have to do to get the right weight and lastly \emph{Update Time} is the time needed to solve each following $(P, k)$-best step including updating the bounds and data structures. 

Note, that we still have to argue about finding the initial solution,
since in the construction the parameters of the \nfold{} slightly change.
The length of a brick changes to $t' = t + r + s$.
This, however, can be hidden in the $\cO$-Notation of $(rs\Delta)^{\cO(r^2s + s^2)}$.
Further, $\Gamma'$, the biggest difference in upper and lower bounds
can be bounded by a function in $\Gamma, \Delta, L, t$ and $n$.
Recall, that the difference between the bounds of old variables does not change. For the
new variables, however, the difference can be as large as $\lVert b' \rVert_\infty$.
Thus we bound this value by
\begin{equation*}
  \lVert b' \rVert_\infty = \lVert b - \mathcal A \ell \rVert_\infty
  \le \lVert b \rVert_\infty + \lVert \mathcal A \ell \rVert_\infty
  \le \lVert b \rVert_\infty + \Delta \cdot \lVert \ell \rVert_1
  \le \cO(\Delta \cdot nt \cdot 2^L) .
\end{equation*}
We conclude that the running time for finding an initial solution (and also the overall running time) is
\begin{multline*}
(rs\Delta)^{\cO(r^2s + s^2)} \cO(\log^2(\Gamma' + 2^L) nt^2 \log^2(nt)) = (rs\Delta)^{\cO(r^2s + s^2)} \cO(\log^2(\Delta 2^L nt) nt^2 \log^2(nt)) \\
= (rs\Delta)^{\cO(r^2s + s^2)}  \cdot L^2 nt \log^4(nt) . \qedhere
\end{multline*}
\end{proof}

\paragraph*{Handling Infinite Bounds}
Remark, that if no finite bounds are given for all variables, we have to introduce some artificial bounds first. Here we can proceed as in \cite{klein}, where first the LP relaxation is solved to obtain an optimal fractional solution $z^*$. Using the proximity results from $\cite{klein}$, we know that an optimal integral solution $x^*$ exists such that $\norm{x^* - z^*}_1 \leq nt(rs\Delta)^{\cO(rs)}$. This allows us to introduce artificial upper bounds for the unbounded variables. Remark that this comes at the price of solving the corresponding relaxation of the \nfold{} Integer Linear Program problem. However we also lessen the dependency from $L^2$ to $L$ as the finite upper and lower bounds can also be bounded more strictly due to the same proximity result. This yields an overall running time of $(rs\Delta)^{\cO(r^2s + s^2)} \cdot L \cdot nt \log^4(nt) + \mathrm{LP}$.
Nevertheless, solving this LP can be very costly, indeed it is not clear if a potential algorithm even runs in time linear in $n$. Thus, it may even dominate the running time of solving the \nfold{} ILP with finite upper bounds. Fortunately we can circumvent the necessity of solving the LP as we will describe in the following section using new structural results.    

\begin{theorem}
The dynamic augmenting step algorithm described above computes an optimal solution for the \nfold{} Integer Linear Program problem in time $(rs\Delta)^{\cO(r^2s + s^2)} \cdot L \cdot nt \log^4(nt) + LP$ when some variables have infinite upper bounds. Here $LP$ is the running time to solve the corresponding relaxation of the \nfold{} ILP problem. 
\end{theorem}

\section{Bounds on $\ell_1$-norm}
\label{sec-bound}
In the following, we prove that even with infinite variable bounds in an \nfold{}
there always exists a solution of small norm (if the \nfold{} has a finite optimum).
Therefore, we can apply the algorithm for finite variable bounds
by replacing every infinite one with this value.
\begin{lemma}\label{bound-feasible}
  If the \nfold{} is feasible and $\overline y$ is some (possibly infeasible) solution
  satisfying the variable bounds, then
  there exists a feasible solution $x$ with
  $\lVert x \rVert_1 \le \cO(rs\Delta)^{rs+1}\cdot (\lVert \overline y \rVert_1 + \lVert b \rVert_1)$
\end{lemma}
\begin{proof}
We take the same construction as in the algorithm for finding a feasible solution in
Section~\ref{sec-alg}. Indeed, this construction was not setup for infinite bounds,
but we consider the straight-forward adaption where infinite bounds simply stay the same.
The useful property is that
an optimal solution for this \nfold{} is a feasible solution for the original \nfold{}.
Recall, the construction has a right-hand side $b'$ with
$\lVert b' \rVert_1 \le \lVert \mathcal A \overline y \rVert_1 + \lVert b \rVert_1$,
the value of $t$ becomes $t' = t + r + s$,
and the objective function $c'$ consists only of the values $\{-1,0,1\}$. 
Moreover, there is a feasible solution $y$ with
$\lVert y \rVert_1 = \lVert b' \rVert_1$.
Let $x^*$ be an optimal solution for this altered \nfold{}
that minimizes $\lVert x^* - y \rVert_1$.
We consider the decomposition into Graver elements
$\sum_{i=1}^{2n(t+r+s)-1} \lambda_i g_i = x^* - y$.
Then $c^{\prime T} g_i > 0$ or $\lambda_i = 0$ for all $i$,
since otherwise $x^* - g_i$ would be a better solution than $x^*$.
It follows that $c^{\prime T} g_i \ge 1$ by discreteness of $c$.
Also, by Theorem~\ref{gravernorm}, $\lVert g_i \rVert_1 \le \cO(rs\Delta)^{rs}$.
Recall that by construction $c^{\prime T} x^* = 0$ and $c^{\prime T} y = - \lVert b' \rVert_1$,
which implies $\sum_{i} \lambda_i \le \lVert b' \rVert_1$.
Therefore,
\begin{multline*}
  \lVert x^* \rVert_1 \le \lVert y \rVert_1 + \lVert \sum_{i} \lambda_i g_i \rVert_1
  \le \lVert b' \rVert_1 + \sum_{i} \lambda_i \lVert g_i \rVert_1 \\
  \le \lVert b' \rVert_1 + \lVert b' \rVert_1 \cdot \cO(rs\Delta)^{rs}
  \le (\lVert b \rVert_1 + \lVert \overline y \rVert_1) \cdot \cO(rs\Delta)^{rs+1} .
\end{multline*}
Here we use that $ \lVert b' \rVert_1 \le \lVert b \rVert_1 + \lVert \mathcal A \overline y \rVert \le \lVert b \rVert_1 + (r + s)\Delta \cdot \lVert \overline y \rVert_1$.
\end{proof}
\begin{lemma}
  If the \nfold{} is bounded and feasible, then
  there exists an optimal solution $x$ with $\lVert x \rVert_1 \le (rs\Delta)^{\cO(rs)}\cdot (\lVert b \rVert_1 + nt\zeta)$, where $\zeta$ denotes the largest absolute value
among all finite variable bounds.
\end{lemma}
\begin{proof}
Clearly there exists a (possibly infeasible) solution $\overline y$ satisfying the bounds with
$\lVert \overline y \rVert_1 \le nt\zeta$.
By the previous lemma we know that there is a feasible solution $y$ with
$\lVert y \rVert_1 \le (rs\Delta)^{\cO(rs)} \cdot (\lVert b \rVert_1 + nt\zeta)$.
Let $x^*$ be an optimal solution of minimal norm.
W.l.o.g. assume that $x^* - y$ has only non-negative entries. If there is a negative entry,
consider the equivalent \nfold{} problem with the corresponding column inverted and its bounds inverted and swapped.

We know that there is a decomposition of $x^* - y$ into
weighted Graver basis elements $\sum_{i=1}^{2nt-1} \lambda_i g_i = x^* - y$.
Since every $g_i$ is sign-compatible with $x^* - y$,
we have that all $g_i$ are non-negative as well.
Furthermore, it holds that $c^T g_i > 0$ or $\lambda_i=0$ for every $g_i$, since otherwise
$x^* - g_i$ would be a solution
of smaller norm with an objective value that is not worse.
Now suppose toward contradiction that there is some $g_i$ where all variables
in $\supp(g_i)$ have infinite upper bounds.
Then the \nfold{} is clearly unbounded, since $y + \alpha \cdot g_i$ is feasible for every $\alpha > 0$
and in this way we can scale the objective value beyond any bound.
Thus, every Graver basis element adds at least the value $1$ to some
finitely bounded variable.
This implies that $\sum_{i} \lambda_i \le \lVert y \rVert_1 + nt\zeta$: If not, then by pigeonhole principle
there is some finitely bounded variable $x^*_j$ with
\begin{equation*}
  x^*_j = y_j + (\sum_{i} \lambda_i g_i)_j > y_j + \zeta + |y_j| \ge \zeta .
\end{equation*}
Since $x^*$ is feasible, this cannot be the case. We conclude,
\begin{multline*}
  \lVert x^* \rVert_1 \le \lVert y \rVert_1 + \sum_{i} \lambda_i \lVert g_i \rVert_1
  \le \lVert y \rVert_1 + \cO(rs\Delta)^{rs} \cdot \sum_{i} \lambda_i \\
  \le \cO(rs\Delta)^{rs} \cdot (\lVert y \rVert_1 + nt\zeta)
  \le (rs\Delta)^{\cO(rs)} \cdot (\lVert b \rVert_1 + nt\zeta) . \qedhere
\end{multline*}
\end{proof}
This yields an alternative approach to solving the LP relaxation, because now we can
simply replace all infinite bounds with $\pm (rs\Delta)^{\cO(rs)} \cdot n t \cdot 2^L$.
Then we can apply the algorithm that works only on finite variable bounds.
The new encoding length $L'$ of the largest integer in the input can be bounded by
\begin{equation*}
L' \le \log((rs\Delta)^{\cO(rs)} \cdot 2^L \cdot n t) \le \cO(rs\cdot \log(rs\Delta) \cdot L \cdot \log(nt)) .
\end{equation*}
This way we obtain the following.
\begin{corollary}
We can compute an optimal solution for an \nfold{} in time
$(rs\Delta)^{\cO(r^2s + s^2)} L^2 \cdot nt \log^6(nt)$.
\end{corollary}
In a similar way, we can derive the following bound on the sensitivity of an \nfold{} ILP.
This bound is not needed in our algorithm, but may be of independent interest, since it
implies small sensitivity for problems that can be expressed as \nfold{}.
 \begin{theorem}
 Let $x$ be an optimal solution of an \nfold{} with right-hand side $b$, in particular, $\mathcal Ax = b$. If the right hand side changes to $b'$ and the \nfold{} still has a finite optimum, then there exists an optimal solution $x'$ for $b'$ ($\mathcal Ax' = b'$) with $\lVert x - x' \rVert_1 \le \cO(rs\Delta)^{rs} \cdot \lVert b - b' \rVert_1$.
 \end{theorem}
It is notable that this bound does not depend on $n$. This is in contrast to the known bounds
for the distance between LP and ILP solutions of an \nfold{}~\cite{klein}.
\begin{proof}
  Consider the matrix $\mathcal A^\text{init}$ from the
  construction used for finding an initial solution, that is,
  identity matrices are added after every block.
  As opposed to
  the proof of Lemma~\ref{bound-feasible}, we leave everything except for the matrix the same.
  In particular, we do not change the value in the objective function $c$
  and new columns get a value of $0$.
  As the right-hand side of the \nfold{} we use $b'$.
  For some solution $x$, we write $x_\text{old}$ and $x_\text{new}$ for the vector restricted
  to the old variables (with all others $0$)
  and the variables added in the matrix $\mathcal A^\text{init}$, respectively.
  This means $x = x_\text{old} + x_\text{new}$.

  Let $x$ be an optimal solution with $\mathcal A^\text{init} \cdot x_\text{new} = b' - b$
  and $x'$ one with $\mathcal A^\text{init} \cdot x'_\text{new} = 0$.
  Here we assume that $x'$ is chosen so as to minimize $\lVert x - x' \rVert_1$.
  Those solutions naturally correspond to solutions of the original \nfold{} with right-hand side
  $b = b' - (\mathcal A^\text{init} \cdot x_\text{new})$ and $b' = b' - \mathcal A^\text{init} \cdot x'_\text{new}$.
  Let $\sum_{i=1}^{2n(t+r+s)-1} \lambda_i g_i = x' - x$ be the decomposition into Graver basis
  elements.
  Suppose toward contradiction there is some $g_i$ where all of $\supp(g_i)$ are old variables.
  If $c^T g_i > 0$, then $x$ is not optimal, because $x + g_i$ is feasible and has a better
  objective value. If on the other hand $c^T g_i \le 0$, then $x' - g_i$ is a solution
  of at least the same value as $x'$ and thus $\lVert x - x' \rVert_1$ is not minimal.
  Indeed, this means $\lVert (g_i)_\text{new} \rVert_1 \ge 1$ for all $g_i$. In other words,
  each graver element contains a non-zero new variable.
  Recall that $\mathcal A^\text{init}$ is the identity matrix when restricted to the new variables
  (plus some zero columns).
  Due to the sign-compatibility we get
  \begin{equation*}
    \sum_i \lambda_i \le
    \lVert (\sum_{i} \lambda_i g_i)_\text{new} \rVert_1
    = \lVert \mathcal A^\text{init} \cdot (\sum_{i} \lambda_i g_i)_\text{new} \rVert_1
    = \lVert \mathcal A^\text{init} \cdot (x' - x)_\text{new} \rVert_1
    = \lVert b - b' \rVert_1 .
  \end{equation*}
  We conclude,
  \begin{equation*}
    \lVert x - x' \rVert_1 = \lVert \sum_{i} \lambda_i g_i \rVert_1 
    \le \cO(rs\Delta)^{rs} \cdot \sum_{i} \lambda_i \le \cO(rs\Delta)^{rs} \cdot \lVert b - b' \rVert_1 .\qedhere
  \end{equation*}
\end{proof}

\bibliography{ref}
\end{document}